\documentclass[]{interact}
\usepackage[scr=rsfs]{mathalpha}
\usepackage{mathtools}
\usepackage{graphicx}
\usepackage{amsmath}
\usepackage{epstopdf}
\usepackage[caption=false]{subfig}
\usepackage{comment}

\usepackage[displaymath, mathlines,running]{lineno}

\theoremstyle{plain}
\newtheorem{assum}{Assumption}
\newtheorem{exmp}{Example}
\newtheorem{defn}{Definition}

\newtheorem{thm}{Theorem}
\newtheorem{rem}{Remark}
\newtheorem{cor}{Corollary}

\usepackage[colorlinks = true,
linkcolor = blue,
urlcolor  = blue,
citecolor = blue,
anchorcolor = blue]{hyperref}

\begin{document}
\global\long\def\R{\mathbb{R}}%
\global\long\def\D{\mathbb{D}}%
\global\long\def\diag{\mathop{\mathrm{diag}}}%
\global\long\def\cl{\mathop{\mathrm{cl}}}%
\global\long\def\he{\mathop{\mathrm{He}}}%
\global\long\def\Oz{\mathbf{O}}%
\global\long\def\supess{\text{sup\,ess}}%
\global\long\def\sign{\mathop{\mathrm{sign}}}%
\global\long\def\ln{\mathop{\mathrm{ln}}}
\global\long\def\mR{\mathbb{R}}%
\global\long\def\L{\mathscr{L}}%
\global\long\def\KL{\mathscr{KL}}%
\global\long\def\K{\mathscr{K}}%

\articletype{ARTICLE TEMPLATE}

\title{On annular short-time stability conditions for generalized Persidskii systems}

\author{\name{Wenjie Mei \textsuperscript{a}\thanks{CONTACT Wenjie Mei. Email: wenjie.mei@inria.fr}, Denis Efimov \textsuperscript{a}\thanks{CONTACT Denis Efimov. Email: denis.efimov@inria.fr},  and Rosane Ushirobira \textsuperscript{a}\thanks{CONTACT Rosane Ushirobira. Email: rosane.ushirobira@inria.fr}} \affil{\textsuperscript{a}Inria, Univ. Lille, CNRS, UMR 9189 - CRIStAL, F-59000 Lille, France}}

\maketitle

\begin{abstract}
This paper studies the trajectory behavior evaluation for generalized Persidskii systems with an essentially bounded input on a finite time interval. Also, the notions of annular settling and output annular settling for general nonlinear systems are introduced. We propose conditions for annular short-time stability, short-time boundedness with a nonzero initial state, annular settling, and output annular settling for a class of Persidskii systems. These conditions are based on the verification of linear matrix inequalities. An application to recurrent neural networks illustrates the usefulness of the proposed notions and conditions.
\end{abstract}


\section{Introduction}
The stability analysis of dynamical systems is a complicated issue,
especially for nonlinear dynamics with external inputs
\cite{Vidyasagar:81:Springer,Schaft:96:Springer,Khalil2002nonlinear}, for which the Lyapunov theory and the asymptotic stability concepts
\cite{Vidyasagar2002} are commonly used. However, in some scenarios,
analyzing the system behavior at all positive times is unnecessary;
one may focus on the trajectories in a bounded time interval. A
representative example is given by the short-time stability concept,
introduced in the 1950s \cite{Dorato,dorato1961short} to
address the problem that the solutions of a system stay in a given
domain, starting from a bounded set of initial conditions, during a
finite time interval. Various investigations of short-time stability
and stabilization of linear time-varying systems were introduced
\cite{Amato2010,Garcia2009b,Amato2009a}, as well as the robust
short-time stability analysis for linear systems
\cite{Amato2001a}. Further studies in distinct directions include
techniques of stochastic systems \cite{Yang2009} or annular
short-time stability \cite{Amato2016a} (called annular
finite-time stability \cite{Amato2016a}), to mention several
examples.

In this work, we call the mentioned stability concept short-time \cite{dorato1961short} instead of finite-time stability (also frequently used to denote the same property
\cite{Dorato}). The reason is that there is another notion of
stability with the same name, dealing with a finite time of
convergence of the system trajectories to an invariant mode
(\emph{e.g.}, an equilibrium or a desired set) combined with Lyapunov stability for all positive times. Nevertheless, all analysis is carried out in a bounded time interval for short-time stability notion.

 We study in this note several concepts. Following \cite{Amatoa,Amato2016a}, robust annular short-time stability (ASTS) and its extensions to short-time boundedness with a nonzero initial state (STBNZ) are investigated, while annular settling (AS) and output annular settling (oAS) are first time introduced. Roughly speaking, a system is said to be ASTS if the system state stays in a bounded domain during
a specified time interval, given an initial bounded region separated from the origin, while in AS case, only boundedness at the end of the time interval is required. The notion of ASTS was introduced for safety
problems, where the state must behave between given lower and upper bounds; for instance, the water level in a tank should not exceed the prescribed thresholds \cite{Amato2016a}.

The dynamical models considered in this work belong to the class of so-called generalized Persidskii systems (extended from the
dynamics in \cite{Barbashin1961,Persidskii1969}), which have been
extensively studied in the context of neural networks
\cite{Ferreira2005}, biological models \cite{Mei2021b}, and power
systems \cite{Hsu1987}.  Recent advances for generalized Persidskii
models include, for example, the conditions of input-to-state
stability, input-to-output stability and convergence, and the
synthesis of a state observer
\cite{Efimov2019a,Mei2020a,Mei2020b,Mei2021}.  Note that the most
existing approaches to synthesizing Lyapunov functions for stability
analysis in nonlinear systems involve various canonical forms of the
studied differential equations, \emph{e.g.}, Lur'e systems
\cite{Guiver2020a}, homogeneous models \cite{Moreno2012}, Persidskii
systems \cite{Persidskii1969}, and Lipschitz dynamics. Due to the
intrinsic nature of nonlinearities, the relevant stability conditions
may be rather sophisticated. However, in generalized Persidskii
systems, a Lyapunov function is found, where the stability conditions can be formulated in the form of linear matrix inequalities (LMIs) \cite{Efimov2019a},
which is a rare case for nonlinear dynamics.

Continuous-time recurrent neural networks constitute an example of
generalized Persidskii systems.  An important application of these neural networks consists in classifying temporal sequences defined on a finite time interval. In such a scenario, the output of a trained network has to
approach desired levels (typically, a small compact set) in the
considered time interval, indicating a class of the input signals. As
we will show such a behavior can be quantified using the oAS notion (if the output equals the state, the
concept of AS can be utilized) that is introduced in this paper. Therefore, the main
contributions of this work include the introduction of these new
notions: oAS and AS for general nonlinear systems, the formulation of
ASTS, STBNZ, AS, and oAS conditions for a family of generalized
Persidskii systems, as well as the illustration of the usefulness and efficacy of the
proposed conditions in recurrent neural networks.

The organization of the rest of this paper is as follows. 
Section~\ref{sec:Prelim} introduces the preliminaries and the definitions of
considered stability properties. The generalized
Persidskii system is presented in Section~\ref{sec:Problem}. In
Section~\ref{sec:ASTS_condition_Per}, ASTS, STBNZ, AS, and oAS
conditions for the class of considered systems are given, and in
Section~\ref{sec:example}, an application to continuous-time recurrent
neural networks is investigated to examine the efficiency of the
proposed results.

\section*{Notation}
\begin{itemize}
\item $\mathbb{N}$, $\mathbb{R}$ and $\R_{+}$ represent the sets of
  natural, real and nonnegative real numbers, respectively.

\item $\mathbb{R}^{n}$ and $\mathbb{R}^{m\times n}$ denote the vector
  spaces of real $n$-vectors and $m \times n$ real matrices,
  respectively. The set of $n \times n$ diagonal matrices (with
  nonnegative diagonal elements) is denoted by $\mathbb{D}^{n}$
  ($\mathbb{D}^{n}_{+}$).

\item $I_{n}$ and $O_{m\times n}$ stand for the identity $n\times
  n$-matrix and the zero $m\times n$-matrix, respectively. Denote the
  vector of dimension $n$ with all elements equal $1$ by ${\bf{1}}_n$.

\item The symbol $\rVert\cdot\rVert$ refers to the Euclidean norm on
  $\mathbb{R}^{n}$ (and the induced matrix norm $\rVert A\rVert$ for a
  matrix $A\in\mathbb{R}^{m\times n}$).
  \item For $p$, $n\in\mathbb{N}$ with $p\leq n$, the notation $\overline{p,n}$
is used to represent the set $\{p,\dots,n\}$.  

\item $A^{(s)}$ stands for the $s$th row of the matrix $A \in
  \mathbb{R}^{m \times n}$, $s \in \overline{1,m}$. For all
  $i,j\in\overline{p,n}$, let $(A_{i,j})_{i,j=p}^{n}$ denote the block
  matrix {\scriptsize $\begin{bmatrix} A_{p,p} & \cdots &
      A_{p,n}\\ \vdots & \ddots & \vdots\\ A_{n,p} & \cdots & A_{n,n}
\end{bmatrix}$}. 

\item $B_n(\delta) = \{y\in \R^n: \; \rVert y \rVert \leq \delta \}$
  stands for the closed ball of radius $\delta>0$ centered at the
  origin; $\cl(E)$ denotes the closure of a set $E \subset \R^n$.
\item For a Lebesgue measurable function
  $u\colon\mathbb{R}\rightarrow\mathbb{R}^{m}$, define the norm
  $\rVert
  u\rVert_{(t_{1},t_{2})}=\text{ess}\sup_{t\in(t_{1},t_{2})}\rVert
  u(t)\rVert$ for $(t_{1},t_{2})\subseteq \mathbb{R}$. Denote by
  $\mathscr{L}_{(t_1,t_2)}^{m}$ (or $\mathscr{L}_{\infty}^{m}$) the
  Banach space of functions $u$ with $\rVert
  u\rVert_{(t_1,t_2)}<+\infty$ (or $\rVert u\rVert_{\infty}:=\rVert
  u\rVert_{(-\infty,+\infty)}<+\infty$).
\item A continuous function $\sigma:\mathbb{R}_{+}\to\mathbb{R}_{+}$ belongs
to class $\mathscr{K}$ if it is strictly increasing and $\sigma(0)=0$;
it belongs to class $\mathscr{K}_{\infty}$ if it is also unbounded.
\item For a continuously differentiable function $V\colon\mathbb{R}^{n}\to\mathbb{R}$,
denote by $\nabla V(\nu)f(\nu)$ the derivative of $V$ along
the vector field $f$ evaluated at point $\nu\in\mathbb{R}^{n}$. 
\item The maximum eigenvalue of a symmetric matrix $P$ is denoted by $\lambda_{\max}(P)$. 
\end{itemize}

\section{Preliminaries} \label{sec:Prelim}
Consider the differential equation 
\begin{gather}
\begin{aligned}
\dot{x}(t) & =F(x(t),u(t)),\quad t \in  \Delta = [0,T] \subset \R, \\
y(t) & = h(x(t)),
\end{aligned}
\label{eq:general_nonlinear_system}
\end{gather}
where $x(t)\in\mathbb{R}^{n}$ is the state vector;
$u(t)\in\mathbb{R}^{m}$ is the external input,
$u\in\mathscr{L}_{\infty}^{m}$; $\Delta$ is the time interval of
interest with $0<T<+\infty$. Moreover,
$F\colon\mathbb{R}^{n}\times\mathbb{R}^{m}\rightarrow\mathbb{R}^{n}$
is a continuous function and
$h\colon\mathbb{R}^{n}\rightarrow\mathbb{R}^{p}$ is a continuously
differentiable function.  In the rest of the paper, to lighten the
notation, the time-dependency of variables might remain implicitly
understood; for instance we will write $x$ for $x(t)$.

For an initial state $x(0)=x_{0}\in\mathbb{R}^{n}$ and
$u\in\mathscr{L}_{\infty}^{m}$, we denote the corresponding solution
of system~\eqref{eq:general_nonlinear_system} by $x(t,x_{0},u)$
for the values of $t\in\mathbb{R_{+}}$ the solution exists, so the
corresponding output is $y(t,x_0,u)=h(x(t,x_0,u))$. In the sequel, we
assume that such a solution of~\eqref{eq:general_nonlinear_system} is
uniquely defined, for any $x_0\in\R^n$ and $u\in\L_\infty^m$, for all
$t\in\Delta$.

\begin{defn}[\cite{Amatoa,Amato2016a}]
\label{def:IOS} System~\eqref{eq:general_nonlinear_system}
is said to be:
\begin{enumerate}
\item {\em short-time bounded} with respect to
  $(\Delta,\gamma,\delta)$, if for given $\gamma \geq 0$ and $\delta
  >0$, 
\[
x_0 =0, \rVert u \rVert_{\Delta} \leq \gamma  \quad  \Rightarrow \quad    \rVert x(t,x_0,u) \rVert \leq \delta, \quad \forall t \in \Delta.
\] 
\item {\em short-time bounded with nonzero initial state} (STBNZ) with
  respect to $(\Delta,\epsilon,\gamma,\delta)$, if for given $\epsilon
  >0$, $\gamma \geq 0$ and $\delta >0$,
\[
\|x_0\|\leq\varepsilon,\; \rVert u \rVert_{\Delta} \leq \gamma \quad  \Rightarrow  \quad  \rVert x(t,x_0,u) \rVert \leq \delta, \quad \forall t \in \Delta.
\] 
\item {\em annular short-time stable} (ASTS) with respect to
  $(\Delta,\epsilon_1,\epsilon_2,\gamma,\delta_1,\delta_2)$, if for
  given $(\epsilon_1, \epsilon_2) \subseteq (\delta_1, \delta_2)
  \subset \mathbb{R}_+$ and $\gamma \geq 0$,
\begin{gather*}
x_0 \in \cl(B_n(\epsilon_2) \setminus B_n(\epsilon_1)), \; \rVert u \rVert_{\Delta} \leq \gamma  \\
\Rightarrow  \quad x(t,x_0,u) \in \cl(B_n(\delta_2) \setminus B_n(\delta_1)), \quad \forall t \in \Delta.
\end{gather*}
\end{enumerate}
\end{defn}

\medskip

In the classification of temporal sequences by neural networks,  the resulting output level at the time instant $T$ is  considered, whose value characterizes the class to which the given input signal belongs  (a complete analysis  over  the whole time interval $\Delta$ is nevertheless  unnecessary in this scenario, as it is considered in the case of ASTS). For investigation of this behavior, we further introduce two useful notions:

\begin{defn} \label{def:AS_oAS}
System~\eqref{eq:general_nonlinear_system} is said to be:
\begin{enumerate}
\item {\em annular settled} (AS) with respect to
  $(T,\epsilon_1,\epsilon_2,\gamma,\delta_1,\delta_2)$, if for given
  $(\epsilon_1, \epsilon_2) \subset \mathbb{R}_+$, $(\delta_1,
  \delta_2) \subset \mathbb{R}_+$ and $\gamma \geq 0$,
\begin{gather*}
x_0 \in \cl(B_n(\epsilon_2) \setminus B_n(\epsilon_1)), \; \rVert u \rVert_{\Delta} \leq \gamma  \\
\Rightarrow  \quad x(T,x_0,u) \in \cl(B_n(\delta_2) \setminus B_n(\delta_1)).
\end{gather*}
\item \em output annular settled (oAS) with respect to
$(T,x_0,\gamma,\delta_1,\delta_2)$, if for given   $x_0\in\R^n$, $(\delta_1, \delta_2) \subset \mathbb{R}_+$ and $\gamma \geq 0$,
\begin{gather*}
\rVert u \rVert_{\Delta} \leq \gamma  \Rightarrow  \quad y(T,x_0,u) \in \cl(B_p(\delta_2) \setminus B_p(\delta_1)).  
\end{gather*}
\end{enumerate} 
\end{defn}

Therefore, all considered properties are not related to the attractiveness of a set, but with visiting this set by trajectories during the interval of time $\Delta$ for ASTS or just at the end of the interval for AS and oAS. Next, the trajectories may leave the sets of interest.

\section{Problem Statement} \label{sec:Problem}

Consider the following system in a generalized Persidskii form
\cite{Persidskii1969}:
\begin{equation}
\begin{aligned}
\dot{x}(t) &=A_{0}x(t)+\sum_{j=1}^{M}A_{j}F_{j}(H_j x(t))+u(t),\;t\in [0,T], \\
y(t) &= Cx(t), \label{eq:main_system_Per} 
\end{aligned}
\end{equation}
where $x=[\begin{array}{ccc} x_{1} & \ldots &
    x_{n}\end{array}]^{\top}\in\mathbb{R}^{n}$ is the state,
$x(0)=x_0$; $y(t)\in\R^p$ is the output signal; $[0,T]$ is the
interval of interest for some $0<T<+\infty$;
$u\in\mathscr{L}_{\infty}^{n}$ is the exogenous input; $A_0\in
\mathbb{R}^{n \times n}$, $A_j\in \mathbb{R}^{n \times k_j}$, $H_j \in
\mathbb{R}^{k_j \times n}$ $\left(j \in\overline{1,M} \right)$, $C \in
\R^{p \times n}$ are constant matrices (to shorten further writing, we
define $k_0=n$ and $H_0=I_n$); $F_j:\mathbb{R}^{k_j} \to
\mathbb{R}^{k_j}$ are continuous functions,
\begin{gather*}
F_{j}(\ell^j)=[\begin{array}{ccc} f_{j}^{1}(\ell_{1}) & \ldots &
    f_{j}^{k_j}(\ell_{k_j})\end{array}]^{\top}\in\mathbb{R}^{k_j},\;\forall
j\in\overline{1,M}
\end{gather*}
with $\ell^j =[\begin{array}{ccc} \ell_{1} & \ldots &
    \ell_{k_j} \end{array}]^\top \in \R^{k_j}$.  Thus, the
nonlinearity $F_{j}$ has a special structure: each component
$f_{j}^{i}:\mathbb{R} \to \mathbb{R}$ of $F_{j}$ depends only on
$H_j^{(i)}x$, for $i \in\overline{1,k_j}$.

Sector restrictions on $F_{j} \left(j\in\overline{1,M} \right)$ are imposed in
the following assumption:
\begin{assum} \label{assu:main} 
Assume that for any  $j\in\overline{1,M}$ and $i \in\overline{1,k_j}$
\begin{equation*}
\nu f_{j}^{i}(\nu)>0,\ \forall\nu\in\mathbb{R}\backslash\{0\}.
\end{equation*}
\end{assum} 
\vspace{\baselineskip}

Under Assumption~\ref{assu:main}, with a reordering of nonlinearities and their
decomposition, there exists an index $c \in\overline{0,M}$ such that
for all $a\in\overline{1,c}$, $i\in\overline{1,k_a}$
\[
\lim_{\nu\rightarrow\pm\infty}f_{a}^{i}(\nu)=\pm\infty,
\]
and that there exists $\mu\in\overline{c,M}$ such that for all 
$b\in\overline{1,\mu}$, $i\in\overline{1,k_b}$
\[
\lim_{\nu\rightarrow\pm\infty}\int_{0}^{\nu}f_{b}^{i}(\tau)d \tau=+\infty.
\]
In this case, $c=0$ implies that all nonlinearities are bounded (at least for positive or negative argument). Clearly, $\mu \geq c$. 

\medskip

In this study, to perform the stability analysis, we also need a mild assumption of upper and
lower bounds on the integrals of the nonlinearities:

\begin{assum} \label{assu:integral_bound} Assume that for any 
$j$, $ j'\in\overline{1,M}$, $z\in\overline{j'+1,M}$, $i
  \in\overline{1,k_{j}}$, $i' \in\overline{1,n}$ and
  $\Lambda_j=\diag(\Lambda_j^1,...,\Lambda_j^{k_j}) \in \D_+^{k_j}$,
  there exist $\kappa_{0,j}^{i}$, $\kappa_{1,jj'}^{i}$,
  $\kappa_{2,jj'}^{i}$, $\kappa_{3,jj'z}^{i'}$, $\eta_{0,j}^{i}$,
  $\eta_{1,jj'}^{i}$, $\eta_{2,jj'}^{i}$, $\eta_{3,jj'z}^{i'} \geq0$,
  such that
\begin{gather*}
\mathrm{x}^\top H_j^\top \kappa_{0,j} H_j \mathrm{x}+\sum_{j'=1}^{M}f_{j'} (H_{j'} \mathrm{x})^\top \Bigg(\kappa_{1,jj'} f_{j'}(H_{j'}\mathrm{x}) + 2\kappa_{2,jj'} H_{j'} \mathrm{x}\Bigg) \\ 
+2 \sum_{j'=1}^{M} \sum_{z=j'+1}^{M} f_{j'}(H_{j'}\mathrm{x})^\top H_{j'} \cdot \kappa_{3,jj'z}  \cdot H_z^\top f_{z}(H_z\mathrm{x}) \\
\leq  \quad 2 {\bf{1}}_{k_j}^\top \Lambda_j \begin{bmatrix}  \int_{0}^{H_j^{(1)} \mathrm{x}}f_{j}^{1}(s)ds \\ \vdots \\ \int_{0}^{H_j^{(k_j)} \mathrm{x}}f_{j}^{k_j}(s)ds \end{bmatrix} \quad  \leq  \\
\mathrm{x}^\top H_j^\top \eta_{0,j} H_j  \mathrm{x}+\sum_{j'=1}^{M}f_{j'} (H_{j'} \mathrm{x})^\top \Bigg(\eta_{1,jj'} f_{j'}(H_{j'}\mathrm{x}) + 2\eta_{2,jj'} H_{j'} \mathrm{x}\Bigg) \\ 
+2 \sum_{j'=1}^{M} \sum_{z=j'+1}^{M} f_{j'}(H_{j'}\mathrm{x})^\top H_{j'} \cdot \eta_{3,jj'z}  \cdot H_z^\top f_{z}(H_z\mathrm{x})
\end{gather*}
for all $\mathrm{x} 
\in\mathbb{R}^{n}$, where \begin{gather*}
\kappa_{0,j}=\mathrm{diag}(\kappa_{0,j}^{1},...,\kappa_{0,j}^{k_j}),\;\kappa_{1,jj'}=\mathrm{diag}(\kappa_{1,jj'}^{1},...,\kappa_{1,jj'}^{k_j}), 
\\ \kappa_{2,jj'}=\mathrm{diag}(\kappa_{2,jj'}^{1},...,\kappa_{2,jj'}^{k_j}),\;\kappa_{3,jj'z}=\mathrm{diag}(\kappa_{3,jj'z}^{1},...,\kappa_{3,jj'z}^{n}), \\
\eta_{0,j}=\mathrm{diag}(\eta_{0,j}^{1},...,\eta_{0,j}^{k_j}),\;\eta_{1,jj'}=\mathrm{diag}(\eta_{1,jj'}^{1},...,\eta_{1,jj'}^{k_j}),
\\ \eta_{2,jj'}=\mathrm{diag}(\eta_{2,jj'}^{1},...,\eta_{2,jj'}^{k_j}),\;\eta_{3,jj'z}=\mathrm{diag}(\eta_{3,jj'z}^{1},...,\eta_{3,jj'z}^{n}).
\end{gather*} 
\end{assum}
\vspace{1em}

This hypothesis is satisfied by many nonlinear functions: sigmoid functions in neural networks; 
for
polynomials, for example, it is sufficient to select
$\kappa_{2,jj'}\ne0$ and $\eta_{2,jj'}\ne0$. 

In this work, if the upper bound is smaller than the lower one for an index, then the corresponding term (in sum or a sequence) must be omitted.

\section{Stability conditions} \label{sec:ASTS_condition_Per}

In this section, ASTS, STBNZ, AS and oAS sufficient conditions for the
generalized Persidskii system~\eqref{eq:main_system_Per} in the
presence of an essentially bounded input are formulated.

Following \cite{Efimov2021,Mei2020a,Mei2020b,Mei2021}, the stability
analysis of (2) can be performed using a Lyapunov function
\[
V(x)=x^\top Px+2\sum_{j \in \overline{1,M}}\sum_{i \in \overline{1,k_j}}\Lambda_j^i \int_{0}^{H_j^{(i)} x} f_j^i(\nu)d \nu,
\]
where $0 \leq P = P^\top\in\R^{n\times n}$ and
$\Lambda_j=\diag(\Lambda_j^1,...,\Lambda_j^n) \in \D_+^{k_j}$ are
tuning matrices. If they are selected in a way ensuring positive
definiteness of $V$ under Assumption~\ref{assu:main}, then there exist
$\alpha_1^{P,\Lambda_1,\dots,\Lambda_M},
\alpha_2^{P,\Lambda_1,\dots,\Lambda_M}\in\mathcal{K}_\infty$ such that
\[\alpha_1^{P,\Lambda_1,\dots,\Lambda_M}(\|x\|)\leq V(x)\leq\alpha_2^{P,\Lambda_1,\dots,\Lambda_M}(\|x\|)\]
for all $x\in \R^n$. For example, the upper bound can be always taken as
\begin{gather*}
\alpha_2^{P,\Lambda_1,\dots,\Lambda_M}(\tau) = \lambda_{\max}(P)
\tau^2 + 2 \left( \sum_{j=1}^{M} k_j \right)  \cdot \max_{j \in
  \overline{1,M}, i \in \overline{1,k_j}} \left\{ \Lambda_j^i
\int_{0}^{\rVert H_j^{(i)} \rVert \tau} f_j^i(\nu) \ d\nu \right\}.
\end{gather*}
These functions $\alpha_1^{P,\Lambda_1,\dots,\Lambda_M},
\alpha_2^{P,\Lambda_1,\dots,\Lambda_M}$ will be used later in the
proofs.

The system~\eqref{eq:main_system_Per} is highly nonlinear with
multiple nonlinearities, which makes the analysis of ASTS complicated as those nonlinearities may override the linear part readily and significantly influence the behavior  of $x(t)$. The following theorem presents ASTS conditions for
system~\eqref{eq:main_system_Per}, and it is the principal theoretical
contribution of the paper.

\bigskip

\begin{thm}
\label{theorem:ASTS_main_thm}  Let assumptions~\ref{assu:main} and~\ref{assu:integral_bound} be satisfied and $T >0$; $\gamma_{0} \geq 0$; $(\epsilon_1, \epsilon_2) \subseteq (\delta_1, \delta_2)  \subset \R_{+}$ be given. If there exist  $0 \leq \overline{P} = \overline{P}^\top, \; \underline{P} = \underline{P}^\top \in \R^{n \times n}$; $\left \{\overline{\Lambda}_j = \diag(\overline{\Lambda}_j^1,...,\overline{\Lambda}_j^{k_j}), \; \underline{\Lambda}_j = \diag(\underline{\Lambda}_j^1,...,\underline{\Lambda}_j^{k_j}\right\}_{j=1}^M \subset \mathbb{D}_{+}^{k_j}$; $\left\{\overline{\Upsilon}_{0,j},\; \underline{\Upsilon}_{0,j}\right\}_{j=1}^{M} \subset \mathbb{D}^{k_j}; \left\{\overline{\Upsilon}_{s,z},\; \underline{\Upsilon}_{s,z}\right\}_{1 \leq s < z \leq M} \subset \mathbb{D}^{n}$; $\left\{ \overline{\Omega}^j,\; \underline{\Omega}^j \right\}_{j=1}^M \subset \mathbb{R}^{n \times k_j}$; symmetric matrices $\overline{\Gamma},\; \underline{\Gamma},\;  \overline{\Psi},\; \underline{\Psi}  \in \mathbb{R}^{n \times n}$, $\left\{ \overline{\Xi}^s,\; \underline{\Xi}^s \right\}_{s=0}^M \subset \mathbb{R}^{k_s \times k_s}$; 
$\overline{\gamma}, \; \underline{\gamma}>0$ and  $\beta_1,\beta_2,\overline{\rho},\; \underline{\rho} \in \R$  such that 
\begin{align}
  & \underline{P} + \underline{\rho} \sum_{j=1}^{\mu} H_j^\top \underline{\Lambda}_j H_j >0,  \label{underline_V_estimate} \\
 & \underline{Q}=\underline{Q}^{\top} = \left(\underline{Q}_{a,\,b}\right)_{a,\,b=1}^{M+3} \geq 0,\label{underline_Q_restriction} \\
& 
\begin{cases} \label{gamma_underline_condi}
\underline{\gamma} \leq \frac{ e^{\beta_1T}\underline{\alpha}_1(\epsilon_1) - \underline{\alpha}_2(\delta_1)}{\frac{\gamma_0^2}{\beta_1} \left[ e^{\beta_1 T} -1 \right]}, & \text{if}\; \beta_1 < 0 \; \text{or} \; \beta_1 >0, \underline{\gamma} > \frac{\beta_1 }{\gamma_0^2} \underline{\alpha}_1(\epsilon_1),\\  
 \underline{\gamma} \leq \frac{   \underline{\alpha}_1(\epsilon_1) - \underline{\alpha}_2(\delta_1) }{T\gamma_0^2}, & \text{if}\; \beta_1 = 0,\\ 
  \underline{\alpha}_2(\delta_1)  \leq  \underline{\alpha}_1(\epsilon_1), & \text{if}\; \beta_1 >0, \underline{\gamma} \leq \frac{\beta_1 }{\gamma_0^2} \underline{\alpha}_1(\epsilon_1).
\end{cases} 
\\
&\nonumber \\
&   \begin{cases} \label{eq:beta_one_condi}
 \underline{\Lambda}_{j} \leq \Lambda_{j}, \\
-\underline{\Xi}^{0}\geq \beta_1\left(\underline{P}+\sum_{j=1}^{M}H_j^\top \eta_{0,j} H_j\right),
\\ 
-\underline{\Xi}^{j}\geq \beta_1 \sum_{j'=1}^{M}\eta_{1,jj'},\; & \text{if} \; \beta_1 \geq 0, \;  
\\ 
-\underline{\Upsilon}_{0,j}\geq \beta_1 \sum_{j'=1}^{M}\eta_{2,jj'}, \\
-\underline{\Upsilon}_{j,z} \geq \beta_1\sum_{j'=1}^{M}  \eta_{3,jj'z}, \\ \\
 \underline{\Lambda}_{j} \geq \Lambda_{j}, \\
 -\beta_1\left(\underline{P}+\sum_{j=1}^{M} H_j^\top \kappa_{0,j} H_j  \right) \geq\underline{\Xi}^{0}, \\ 
-\beta_1 \sum_{j'=1}^{M}\kappa_{1,jj'} \geq \underline{\Xi}^{j},\; & \text{if} \; \beta_1<0, \; 
\\ 
-\beta_1 \sum_{j'=1}^{M}\kappa_{2,jj'} \geq \underline{\Upsilon}_{0,j},
\\ 
-\beta_1 \sum_{j'=1}^{M}  \kappa_{3,jj'z} \geq \underline{\Upsilon}_{j,z}.
\end{cases}  
\end{align}
\begin{align}
&\overline{P} + \overline{\rho} \sum_{j=1}^{\mu}
  H_j^\top \overline{\Lambda}_j H_j >0, \label{overline_V_estimate} \\ &
  \overline{Q}=\overline{Q}^{\top}=
  \left(\overline{Q}_{a,\,b}\right)_{a,\,b=1}^{M+3}
  \leq0, \label{overline_Q_restriction}\\ &
   \begin{cases}\label{gamma_overline_condi} 
   \overline{\gamma} \leq
      \frac{ \overline{\alpha}_1(\delta_2) -e^{\beta_{2} T}
        \overline{\alpha}_2(\epsilon_2) }{\frac{\gamma_0^2}{\beta_2}
        \left[ e^{\beta_2 T} -1 \right]}, & \text{if} \; \beta_2 > 0 \; \text{or} \; \beta_2 <0,
      \overline{\gamma} > -\frac{\beta_2}{\gamma_0^2}, \\ 
      \overline{\gamma} \leq \frac{ \overline{\alpha}_1(\delta_2)
        -\overline{\alpha}_2(\epsilon_2) }{T\gamma_0^2}, & \text{if}\;
      \beta_2 = 0, \\ \overline{\alpha}_2(\epsilon_2) \leq
      \overline{\alpha}_1(\delta_2), & \text{if}\; \beta_2 <0,
      \overline{\gamma} \leq -\frac{\beta_2}{\gamma_0^2}
      \overline{\alpha}_2(\epsilon_2).
\end{cases} 
\\ &\nonumber
\end{align}
\begin{align}
& \begin{cases} \label{eq:beta_two_condi}
 \overline{\Lambda}_{j} \leq \Lambda_{j}, \\
  \overline{\Xi}^{0}\geq
  -\beta_2\left(\overline{P}+\sum_{j=1}^{M} H_j^\top \eta_{0,j} H_j \right),
  \\ \overline{\Xi}^{j}\geq
  -\beta_2 \sum_{j'=1}^{M}\eta_{1,jj'},\;
  & \text{if} \; \beta_2<0, \\ \overline{\Upsilon}_{0,j}\geq
  -\beta_2 \sum_{j'=1}^{M}\eta_{2,jj'},
  \\ \overline{\Upsilon}_{j,z}  \geq
  -\beta_2  \sum_{j'=1}^{M}  \eta_{3,jj'z}, \\  \\
   \overline{\Lambda}_{j} \geq \Lambda_{j}, \\
    \beta_2\left(\overline{P}+\sum_{j=1}^{M}  H_j^\top \kappa_{0,j} H_j \right)
  \geq -\overline{\Xi}^{0},
  \\ \beta_2 \sum_{j'=1}^{M}\kappa_{1,jj'} \geq
  -\overline{\Xi}^{j},\ & \text{if} \, \beta_2
  \geq 0,
  \\ \beta_2 \sum_{j'=1}^{M}\kappa_{2,jj'} \geq
  -\overline{\Upsilon}_{0,j},
  \\ \beta_2   \sum_{j'=1}^{M}  \kappa_{3,jj'z}  \geq
  - \overline{\Upsilon}_{j,z}. 
\end{cases} 
\end{align}
where
\begin{gather*}
\underline{\alpha}_1(s)=\alpha_1^{\underline{P},\underline{\Lambda}_1,\dots,\underline{\Lambda}_M}(s), \; \underline{\alpha}_2(s)=\alpha_2^{\underline{P},\underline{\Lambda}_1,\dots,\underline{\Lambda}_M}(s), \\
\overline{\alpha}_1(s)=\alpha_1^{\overline{P},\overline{\Lambda}_1,\dots,\overline{\Lambda}_M}(s), \;
\overline{\alpha}_2(s)=\alpha_2^{\overline{P},\overline{\Lambda}_1,\dots,\overline{\Lambda}_M}(s). \\ \\
\overline{Q}_{1,1}=-\overline{\Psi}^\top - \overline{\Psi}; \quad \overline{Q}_{1,2}= \overline{\Psi}^\top A_0+\overline{P}  -\overline{\Gamma}; \\
\overline{Q}_{1,j+2}=  \overline{\Omega}_j + H_j^\top \overline{\Lambda}_j + \overline{\Psi}^\top A_j; \quad \overline{Q}_{1,M+3} = \overline{\Psi}^\top; \\
\quad \overline{Q}_{2,2} = \overline{\Gamma}^\top A_0 + A_0^\top \overline{\Gamma} +  \overline{\Xi}_0; \quad \overline{Q}_{2,j+2}= \overline{\Gamma}^\top A_j +H_j^\top \overline{\Upsilon}_{0,j} -A_0^\top \overline{\Omega}_j, \\
 \overline{Q}_{2,M+3}= \overline{\Gamma}^\top; 
\quad \overline{Q}_{j+2,j+2}= -\overline{\Omega}_j^\top A_j-A_j^\top\overline{\Omega}_j + \overline{\Xi}^j, \\
\overline{Q}_{s+2,z'+2}=-\overline{\Omega}_s^\top A_{z'} -A_s^\top\overline{\Omega}_{z'} + H_s \overline{\Upsilon}_{s,z'} H_{z'}^\top, \\
\overline{Q}_{j+2,M+3} = -\overline{\Omega}_j^\top; \quad \overline{Q}_{M+3,M+3} = -\overline{\gamma} I_n. \\ \\ 
\underline{Q}_{1,1}=-\underline{\Psi}^\top - \underline{\Psi}; \quad \underline{Q}_{1,2}= \underline{\Psi}^\top A_0+ \underline{P}  - \underline{\Gamma}; \\
\underline{Q}_{1,j+2}=  \underline{\Omega}_j + H_j^\top \underline{\Gamma}^\top + \underline{\Psi}^\top A_j; \quad \underline{Q}_{1,M+3} = \underline{\Psi}^\top; \\
\quad \underline{Q}_{2,2} = \underline{\Gamma}^\top A_0 + A_0^\top \underline{\Gamma}  +  \underline{\Xi}_0; \quad \underline{Q}_{2,j+2}= \underline{\Gamma}^\top A_j +H_j^\top \underline{\Upsilon}_{0,j} -A_0^\top \underline{\Omega}_j; \\
\quad \underline{Q}_{2,M+3}= \underline{\Gamma}^\top; 
\quad \underline{Q}_{j+2,j+2}= -\underline{\Omega}_j^\top A_j -A_j^\top\underline{\Omega}_j +\underline{\Xi}^j, \\
\underline{Q}_{s+2,z'+2}=-\underline{\Omega}_s^\top A_{z'} -A_s^\top\underline{\Omega}_{z'} + H_s \underline{\Upsilon}_{s,z'} H_{z'}^\top, \\
\underline{Q}_{j+2,M+3} = -\underline{\Omega}_j^\top; \quad \underline{Q}_{M+3,M+3} = \underline{\gamma} I_n. 
\end{gather*}
\begin{gather*}
j,j'\in\overline{1,M};\; z\in\overline{j+1,M};  \; s\in\overline{1,M-1};\;z'\in\overline{s+1,M}, 
\end{gather*}
then  the system~\eqref{eq:main_system_Per} is ASTS with respect to $([0,T],\epsilon_1,\epsilon_2,\gamma_0,\delta_1,\delta_2)$.  
\end{thm}

The idea of the proof of this theorem is to consider two Lyapunov functions, $\overline{V}$
and $\underline{V}$, whose upper and lower estimates, respectively, provide the corresponding bounds for the behavior of the state $\|x(t)\|$ (this explains why there are two sets of conditions).

Considering the computation of expressions for $\dot{\overline{V}}$
and $\dot{\underline{V}}$ in the proof of
Theorem~\ref{theorem:ASTS_main_thm}, the term $\dot{x}$ in
$(\dot{x}^{\top} \overline{P} x+x^{\top} \overline{P} \dot{x})$ can
also be expanded using~\eqref{eq:main_system_Per}. As a result, the
terms $A_0^\top \underline{P}+\underline{P}A_0$ and $A_0^\top
\overline{P}+\overline{P}A_0$ will appear in the elements
$\underline{Q}_{2,2}$ and $\overline{Q}_{2,2}$, respectively (together
with other corresponding modifications). Depending on the properties
of $A_0$, such a substitution may provide more possibilities for the
LMIs solution in Theorem \ref{theorem:ASTS_main_thm}.

One shall note that the conditions \eqref{underline_V_estimate}--\eqref{eq:beta_one_condi} and the ones \eqref{overline_V_estimate}--\eqref{eq:beta_two_condi}  are independent of each other, which means that the computational complexity of the LMIs is not high comparing with their variations successfully solved in the literature as in \cite{mei2022nonlinear}, for example. Also, in the given conditions, the selections of the parameters $\epsilon_1, \epsilon_2, \delta_1, \delta_2$ are per the practical demands on the behavior of the solution of the system~\eqref{eq:main_system_Per}, $\beta_1$ and $\beta_2$ can take
positive or negative values, and $\gamma_0$ should be not smaller than $\rVert
  u\rVert_{[0,T]}$ (for practical verification, one can select them based on error and trial method). Taking into account the general shape of the functions
$\underline{V}$ and $\overline{V}$ used for establishing stable or
unstable bounding compartments, one shall see that the corresponding conditions are not restrictive and under proper selections of the tuning parameters, the given
LMIs can be satisfied.

For the formulation of the STBNZ conditions
for system~\eqref{eq:main_system_Per}, we set $\epsilon_1 = \delta_1 =0
$ in Theorem~\ref{theorem:ASTS_main_thm} and obtain the following
corollary:

\begin{cor} \label{cor:STBNZ_condi}
Assume that all conditions of Theorem~\ref{theorem:ASTS_main_thm} are
satisfied under the substitutions $\epsilon_1 \rightarrow 0, \delta_1
\rightarrow 0$, and the eliminations
of~\eqref{underline_V_estimate}--\eqref{eq:beta_one_condi}. Then
system~\eqref{eq:main_system_Per} is STBNZ with respect to
$\left([0,T],\epsilon_2,\gamma_0,\delta_2 \right)$.
\end{cor}

\begin{proof}
For STBNZ property, the conditions~\eqref{underline_V_estimate}--\eqref{eq:beta_one_condi}  under the restrictions of Corollary~\ref{cor:STBNZ_condi} can be omitted due to $\rVert x\rVert \geq 0$ for $ x \in \R^n$.  
\end{proof}

As it follows from definitions \ref{def:IOS} and \ref{def:AS_oAS}, the requirements of AS are much weaker than those of ASTS (the constraints on the state norm are imposed only for $t=T$). At the same time, the former notion is sufficient for investigating the classification problem in neural networks. The following corollary formulates the sufficient conditions for the
AS property of \eqref{eq:main_system_Per}, which will be used in the
next section for an application to recurrent neural networks.
 
\begin{cor} \label{cor_AS}
 Let assumptions~\ref{assu:main} and~\ref{assu:integral_bound} be satisfied and $T >0$; $\gamma_{0} \geq 0$; $ (\epsilon_1, \epsilon_2), \; (\delta_1, \delta_2)   \subset \R_{+}$ be given. If there exist  $0 \leq \overline{P} = \overline{P}^\top, \; \underline{P} = \underline{P}^\top \in \R^{n \times n}$; $\left \{\overline{\Lambda}_j = \diag(\overline{\Lambda}_j^1,...,\overline{\Lambda}_j^{k_j}), \; \underline{\Lambda}_j = \diag(\underline{\Lambda}_j^1,...,\underline{\Lambda}_j^{k_j}\right\}_{j=1}^M \subset \mathbb{D}_{+}^{k_j}$; $\left\{\overline{\Upsilon}_{0,j},\; \underline{\Upsilon}_{0,j}\right\}_{j=1}^{M} \subset \mathbb{D}^{k_j}; \left\{\overline{\Upsilon}_{s,z},\; \underline{\Upsilon}_{s,z}\right\}_{1 \leq s < z \leq M} \subset \mathbb{D}^{n}$; $\left\{ \overline{\Omega}^j,\; \underline{\Omega}^j \right\}_{j=1}^M \subset \mathbb{R}^{n \times k_j}$; symmetric matrices $\overline{\Gamma},\; \underline{\Gamma},\;  \overline{\Psi},\; \underline{\Psi}  \in \mathbb{R}^{n \times n}$, $\left\{ \overline{\Xi}^s,\; \underline{\Xi}^s \right\}_{s=0}^M \subset \mathbb{R}^{k_s \times k_s}$; 
$\overline{\gamma}, \; \underline{\gamma}>0$ and  $\beta_1,\beta_2,\overline{\rho},\; \underline{\rho} \in \R$  such that 
\begin{align}
  & \underline{P} + \underline{\rho} \sum_{j=1}^{\mu} H_j^\top \underline{\Lambda}_j H_j >0,  \nonumber  \\
 & \underline{Q}=\underline{Q}^{\top} = \left(\underline{Q}_{a,\,b}\right)_{a,\,b=1}^{M+3} \geq 0,\nonumber \\
& 
\begin{cases} \underline{\gamma} \leq \frac{ e^{\beta_1T}\underline{\alpha}_1(\epsilon_1) - \underline{\alpha}_2(\delta_1)}{\frac{\gamma_0^2}{\beta_1} \left[ e^{\beta_1 T} -1 \right]},  & \text{if}\; \beta_1 \neq 0,  \\  \underline{\gamma} \leq \frac{   \underline{\alpha}_1(\epsilon_1) - \underline{\alpha}_2(\delta_1) }{T\gamma_0^2}, & \text{if}\; \beta_1 = 0, \end{cases} \nonumber
\\
&   \begin{cases} 
 \underline{\Lambda}_{j} \leq \Lambda_{j}, \\
-\underline{\Xi}^{0}\geq \beta_1\left(\underline{P}+\sum_{j=1}^{M}H_j^\top \eta_{0,j} H_j\right),
\\ 
-\underline{\Xi}^{j}\geq \beta_1 \sum_{j'=1}^{M}\eta_{1,jj'},\; & \text{if} \; \beta_1 \geq 0, \;  
\\ 
-\underline{\Upsilon}_{0,j}\geq \beta_1 \sum_{j'=1}^{M}\eta_{2,jj'}, \\
-\underline{\Upsilon}_{j,z} \geq \beta_1\sum_{j'=1}^{M}  \eta_{3,jj'z}, \\ \\
 \underline{\Lambda}_{j} \geq \Lambda_{j}, \\
 -\beta_1\left(\underline{P}+\sum_{j=1}^{M} H_j^\top \kappa_{0,j} H_j  \right) \geq\underline{\Xi}^{0}, \\ 
-\beta_1 \sum_{j'=1}^{M}\kappa_{1,jj'} \geq \underline{\Xi}^{j},\; & \text{if} \; \beta_1<0, \; \nonumber
\\ 
-\beta_1 \sum_{j'=1}^{M}\kappa_{2,jj'} \geq \underline{\Upsilon}_{0,j},
\\ 
-\beta_1 \sum_{j'=1}^{M}  \kappa_{3,jj'z} \geq \underline{\Upsilon}_{j,z}.
\end{cases}  
\end{align}
\begin{align}
&\overline{P} + \overline{\rho} \sum_{j=1}^{\mu}
  H_j^\top \overline{\Lambda}_j H_j >0, \nonumber  \\ &
  \overline{Q}=\overline{Q}^{\top}=
  \left(\overline{Q}_{a,\,b}\right)_{a,\,b=1}^{M+3}
  \leq0, \nonumber\\ 
  & \begin{cases} \overline{\gamma} \leq
      \frac{ \overline{\alpha}_1(\delta_2) -e^{\beta_{2} T}
        \overline{\alpha}_2(\epsilon_2) }{\frac{\gamma_0^2}{\beta_2}
        \left[ e^{\beta_2 T} -1 \right]}, & \text{if}\;
      \beta_2 \neq 0,   \\  \overline{\gamma} \leq \frac{ \overline{\alpha}_1(\delta_2)
        -\overline{\alpha}_2(\epsilon_2) }{T\gamma_0^2}, & \text{if}\;
      \beta_2 = 0, \end{cases} \nonumber
\\ 
& \begin{cases} \nonumber
 \overline{\Lambda}_{j} \leq \Lambda_{j}, \\
  \overline{\Xi}^{0}\geq
  -\beta_2\left(\overline{P}+\sum_{j=1}^{M} H_j^\top \eta_{0,j} H_j \right),
  \\ \overline{\Xi}^{j}\geq
  -\beta_2 \sum_{j'=1}^{M}\eta_{1,jj'},\;
  & \text{if} \; \beta_2<0, \\ \overline{\Upsilon}_{0,j}\geq
  -\beta_2 \sum_{j'=1}^{M}\eta_{2,jj'},
  \\ \overline{\Upsilon}_{j,z}  \geq
  -\beta_2  \sum_{j'=1}^{M}  \eta_{3,jj'z}, \\  \\
   \overline{\Lambda}_{j} \geq \Lambda_{j}, \\
    \beta_2\left(\overline{P}+\sum_{j=1}^{M}  H_j^\top \kappa_{0,j} H_j \right)
  \geq -\overline{\Xi}^{0},
  \\ \beta_2 \sum_{j'=1}^{M}\kappa_{1,jj'} \geq
  -\overline{\Xi}^{j},\ & \text{if} \, \beta_2
  \geq 0,
  \\ \beta_2 \sum_{j'=1}^{M}\kappa_{2,jj'} \geq
  -\overline{\Upsilon}_{0,j},
  \\ \beta_2   \sum_{j'=1}^{M}  \kappa_{3,jj'z}  \geq
  - \overline{\Upsilon}_{j,z}. 
\end{cases} 
\end{align}
where
\begin{gather*}
\underline{\alpha}_1(s)=\alpha_1^{\underline{P},\underline{\Lambda}_1,\dots,\underline{\Lambda}_M}(s), \; \underline{\alpha}_2(s)=\alpha_2^{\underline{P},\underline{\Lambda}_1,\dots,\underline{\Lambda}_M}(s), \\
\overline{\alpha}_1(s)=\alpha_1^{\overline{P},\overline{\Lambda}_1,\dots,\overline{\Lambda}_M}(s), \;
\overline{\alpha}_2(s)=\alpha_2^{\overline{P},\overline{\Lambda}_1,\dots,\overline{\Lambda}_M}(s). \\ \\
\overline{Q}_{1,1}=-\overline{\Psi}^\top - \overline{\Psi}; \quad \overline{Q}_{1,2}= \overline{\Psi}^\top A_0+\overline{P}  -\overline{\Gamma}; \\
\overline{Q}_{1,j+2}=  \overline{\Omega}_j + H_j^\top \overline{\Lambda}_j + \overline{\Psi}^\top A_j; \quad \overline{Q}_{1,M+3} = \overline{\Psi}^\top; \\
\quad \overline{Q}_{2,2} = \overline{\Gamma}^\top A_0 + A_0^\top \overline{\Gamma} +  \overline{\Xi}_0; \quad \overline{Q}_{2,j+2}= \overline{\Gamma}^\top A_j +H_j^\top \overline{\Upsilon}_{0,j} -A_0^\top \overline{\Omega}_j, \\
 \overline{Q}_{2,M+3}= \overline{\Gamma}^\top; 
\quad \overline{Q}_{j+2,j+2}= -\overline{\Omega}_j^\top A_j-A_j^\top\overline{\Omega}_j + \overline{\Xi}^j, \\
\overline{Q}_{s+2,z'+2}=-\overline{\Omega}_s^\top A_{z'} -A_s^\top\overline{\Omega}_{z'} + H_s \overline{\Upsilon}_{s,z'} H_{z'}^\top, \\
\overline{Q}_{j+2,M+3} = -\overline{\Omega}_j^\top; \quad \overline{Q}_{M+3,M+3} = -\overline{\gamma} I_n. \\ \\ 
\underline{Q}_{1,1}=-\underline{\Psi}^\top - \underline{\Psi}; \quad \underline{Q}_{1,2}= \underline{\Psi}^\top A_0+ \underline{P}  - \underline{\Gamma}; \\
\underline{Q}_{1,j+2}=  \underline{\Omega}_j + H_j^\top \underline{\Lambda}_j^\top + \underline{\Psi}^\top A_j; \quad \underline{Q}_{1,M+3} = \underline{\Psi}^\top; \\
\quad \underline{Q}_{2,2} = \underline{\Gamma}^\top A_0 + A_0^\top \underline{\Gamma}  +  \underline{\Xi}_0; \quad \underline{Q}_{2,j+2}= \underline{\Gamma}^\top A_j +H_j^\top \underline{\Upsilon}_{0,j} -A_0^\top \underline{\Omega}_j; \\
\quad \underline{Q}_{2,M+3}= \underline{\Gamma}^\top; 
\quad \underline{Q}_{j+2,j+2}= -\underline{\Omega}_j^\top A_j -A_j^\top\underline{\Omega}_j +\underline{\Xi}^j, \\
\underline{Q}_{s+2,z'+2}=-\underline{\Omega}_s^\top A_{z'} -A_s^\top\underline{\Omega}_{z'} + H_s \underline{\Upsilon}_{s,z'} H_{z'}^\top, \\
\underline{Q}_{j+2,M+3} = -\underline{\Omega}_j^\top; \quad \underline{Q}_{M+3,M+3} = \underline{\gamma} I_n. 
\end{gather*}
\begin{gather*}
j,j'\in\overline{1,M};\; z\in\overline{j+1,M};  \; s\in\overline{1,M-1};\;z'\in\overline{s+1,M}, 
\end{gather*}  
then  system~\eqref{eq:main_system_Per} is AS with respect to $(T,\epsilon_1,\epsilon_2,\gamma_0,\delta_1,\delta_2)$.  
\end{cor}
\begin{proof}
In such a case, our goal is to ensure the fulfillment of  the relaxed conditions $\overline{V}(x(T)) \leq \overline{\alpha}_{1}(\delta_2), \underline{V}(x(T)) \geq \underline{\alpha}_2(\delta_1)$, which can be directly deduced under the applied substitutions, then the conditions of this corollary imply that all necessary counterparts
from Theorem \ref{theorem:ASTS_main_thm} are verified, and the conclusion follows.
\end{proof}

Note that in this corollary, there is no restriction on relations between $(\epsilon_1,\epsilon_2)$ and $(\delta_1,\delta_2)$. Hence, these intervals may be inside one another or even not intersecting. In applications of neural networks, there is usually an output that has to approach desired levels to classify different temporal sequences. We further consider oAS conditions for system~\eqref{eq:main_system_Per} to address this problem.

\begin{thm}\label{theorem:oAS}  
Let assumptions~\ref{assu:main} and~\ref{assu:integral_bound} be satisfied and $T >0$, $\gamma_{0} \geq 0$, $(\delta_1, \delta_2)  \subset \R_{+}$, and an initial condition $x_0 \in \R^n$ with $Cx_0 \neq 0$ be given. If there exist  $0 \leq \overline{P}_1,\;  \underline{P}_1 \in \R^{n \times n}$, $0 < \overline{P}_2,\; \underline{P}_2 \in \R^{p \times p}$; $\left \{\overline{\Lambda}_j = \diag(\overline{\Lambda}_j^1,...,\overline{\Lambda}_j^{k_j}), \; \underline{\Lambda}_j = \diag(\underline{\Lambda}_j^1,...,\underline{\Lambda}_j^{k_j}\right\}_{j=1}^M \subset \mathbb{D}_{+}^{k_j}$; $\left\{\overline{\Upsilon}_{0,j},\; \underline{\Upsilon}_{0,j}\right\}_{j=1}^{M} \subset \mathbb{D}^{k_j}; \left\{\overline{\Upsilon}_{s,z},\; \underline{\Upsilon}_{s,z}\right\}_{1 \leq s < z \leq M} \subset \mathbb{D}^{n}$; $\left\{ \overline{\Omega}^j,\; \underline{\Omega}^j \right\}_{j=1}^M \subset \mathbb{R}^{n \times k_j}$; symmetric matrices $\overline{\Gamma},\; \underline{\Gamma},\;  \overline{\Psi},\; \underline{\Psi}  \in \mathbb{R}^{n \times n}$, $\left\{ \overline{\Xi}^s,\; \underline{\Xi}^s \right\}_{s=0}^M \subset \mathbb{R}^{k_s \times k_s}$; 
$\overline{\gamma}, \; \underline{\gamma}>0$ and  $\beta_1,\beta_2,\; \underline{\rho},\;  \underline{\ell} \in \R$  such that 
\begin{align}
  & \underline{P} + \underline{\rho} \sum_{j=1}^{\mu} H_j^\top \underline{\Lambda}_j H_j \leq \underline{\ell} C^\top C,   \quad \underline{P} := \underline{P}_1 + C^\top \underline{P}_2 C,  \nonumber \\
 & \underline{Q}=\underline{Q}^{\top} = \left(\underline{Q}_{a,\,b}\right)_{a,\,b=1}^{M+3} \geq 0, \nonumber \\
& 
\begin{cases} \nonumber \underline{\gamma} \leq \frac{ e^{\beta_1T}\underline{\alpha}_1(\rVert Cx_0 \rVert) - \underline{\alpha}_2(\delta_1)}{\frac{\gamma_0^2}{\beta_1} \left[ e^{\beta_1 T} -1 \right]},  & \text{if}\; \beta_1 \neq 0,  \\  \underline{\gamma} \leq \frac{\underline{\alpha}_1(\rVert Cx_0 \rVert) - \underline{\alpha}_2(\delta_1) }{T\gamma_0^2}, & \text{if}\; \beta_1 = 0, \end{cases} 
\\
&\nonumber \\
&   \begin{cases} \nonumber
 \underline{\Lambda}_{j} \leq \Lambda_{j}, \\
-\underline{\Xi}^{0}\geq \beta_1\left(\underline{P}+\sum_{j=1}^{M}H_j^\top \eta_{0,j} H_j\right),
\\ 
-\underline{\Xi}^{j}\geq \beta_1 \sum_{j'=1}^{M}\eta_{1,jj'},\; & \text{if} \; \beta_1 \geq 0, \;  
\\ 
-\underline{\Upsilon}_{0,j}\geq \beta_1 \sum_{j'=1}^{M}\eta_{2,jj'}, \\
-\underline{\Upsilon}_{j,z} \geq \beta_1\sum_{j'=1}^{M}  \eta_{3,jj'z}, \\ \\
 \underline{\Lambda}_{j} \geq \Lambda_{j}, \\
 -\beta_1\left(\underline{P}+\sum_{j=1}^{M} H_j^\top \kappa_{0,j} H_j  \right) \geq\underline{\Xi}^{0}, \\ 
-\beta_1 \sum_{j'=1}^{M}\kappa_{1,jj'} \geq \underline{\Xi}^{j},\; & \text{if} \; \beta_1<0, \; 
\\ 
-\beta_1 \sum_{j'=1}^{M}\kappa_{2,jj'} \geq \underline{\Upsilon}_{0,j},
\\ 
-\beta_1 \sum_{j'=1}^{M}  \kappa_{3,jj'z} \geq \underline{\Upsilon}_{j,z}.
\end{cases}  
\end{align}
\begin{align}
&  \overline{P} := \overline{P}_1 + C^\top \overline{P}_2 C, \nonumber \\ &
  \overline{Q}=\overline{Q}^{\top}=
  \left(\overline{Q}_{a,\,b}\right)_{a,\,b=1}^{M+3}
  \leq0, \nonumber \\ &
   \begin{cases} \overline{\gamma} \leq
      \frac{ \overline{\alpha}_1(\delta_2) -e^{\beta_{2} T}
        \overline{\alpha}_2({\|x_0\|}) }{\frac{\gamma_0^2}{\beta_2} 
        \left[ e^{\beta_2 T} -1 \right]}, & \text{if}\;
      \beta_2 \neq 0,   \\  \overline{\gamma} \leq \frac{ \overline{\alpha}_1(\delta_2)
        -\overline{\alpha}_2({\|x_0\|}) }{T\gamma_0^2}, & \text{if}\; 
      \beta_2 = 0, \end{cases} 
\\ 
& \begin{cases} \nonumber
 \overline{\Lambda}_{j} \leq \Lambda_{j}, \\
  \overline{\Xi}^{0}\geq
  -\beta_2\left(\overline{P}+\sum_{j=1}^{M} H_j^\top \eta_{0,j} H_j \right),
  \\ \overline{\Xi}^{j}\geq
  -\beta_2 \sum_{j'=1}^{M}\eta_{1,jj'},\;
  & \text{if} \; \beta_2<0, \\ \overline{\Upsilon}_{0,j}\geq
  -\beta_2 \sum_{j'=1}^{M}\eta_{2,jj'},
  \\ \overline{\Upsilon}_{j,z}  \geq
  -\beta_2  \sum_{j'=1}^{M}  \eta_{3,jj'z}, \\  \\
   \overline{\Lambda}_{j} \geq \Lambda_{j}, \\
    \beta_2\left(\overline{P}+\sum_{j=1}^{M}  H_j^\top \kappa_{0,j} H_j \right)
  \geq -\overline{\Xi}^{0},
  \\ \beta_2 \sum_{j'=1}^{M}\kappa_{1,jj'} \geq
  -\overline{\Xi}^{j},\ & \text{if} \, \beta_2
  \geq 0,
  \\ \beta_2 \sum_{j'=1}^{M}\kappa_{2,jj'} \geq
  -\overline{\Upsilon}_{0,j},
  \\ \beta_2   \sum_{j'=1}^{M}  \kappa_{3,jj'z}  \geq
  - \overline{\Upsilon}_{j,z}. 
\end{cases} 
\end{align}
where
\begin{gather*}
\underline{\alpha}_1(s)={\lambda_{\min} (\underline{P}_2)s^2}, \; \underline{\alpha}_2(s)=\alpha_2^{\underline{P},\underline{\Lambda}_1,\dots,\underline{\Lambda}_M}(s), \\ 
\overline{\alpha}_1(s)={\lambda_{\min}(\overline{P}_2)s^2}, \; 
\overline{\alpha}_2(s)=\alpha_2^{\overline{P},\overline{\Lambda}_1,\dots,\overline{\Lambda}_M}(s). \\ \\
\overline{Q}_{1,1}=-\overline{\Psi}^\top - \overline{\Psi}; \quad \overline{Q}_{1,2}= \overline{\Psi}^\top A_0+\overline{P}  -\overline{\Gamma}; \\
\overline{Q}_{1,j+2}=  \overline{\Omega}_j + H_j^\top \overline{\Lambda}_j + \overline{\Psi}^\top A_j; \quad \overline{Q}_{1,M+3} = \overline{\Psi}^\top; \\
\quad \overline{Q}_{2,2} = \overline{\Gamma}^\top A_0 + A_0^\top \overline{\Gamma} +  \overline{\Xi}_0; \quad \overline{Q}_{2,j+2}= \overline{\Gamma}^\top A_j +H_j^\top \overline{\Upsilon}_{0,j} -A_0^\top \overline{\Omega}_j, \\
 \overline{Q}_{2,M+3}= \overline{\Gamma}^\top; 
\quad \overline{Q}_{j+2,j+2}= -\overline{\Omega}_j^\top A_j-A_j^\top\overline{\Omega}_j + \overline{\Xi}^j, \\
\overline{Q}_{s+2,z'+2}=-\overline{\Omega}_s^\top A_{z'} -A_s^\top\overline{\Omega}_{z'} + H_s \overline{\Upsilon}_{s,z'} H_{z'}^\top, \\
\overline{Q}_{j+2,M+3} = -\overline{\Omega}_j^\top; \quad \overline{Q}_{M+3,M+3} = -\overline{\gamma} I_n. \\ \\ 
\underline{Q}_{1,1}=-\underline{\Psi}^\top - \underline{\Psi}; \quad \underline{Q}_{1,2}= \underline{\Psi}^\top A_0+ \underline{P}  - \underline{\Gamma}; \\
\underline{Q}_{1,j+2}=  \underline{\Omega}_j + H_j^\top \underline{\Lambda}_j^\top + \underline{\Psi}^\top A_j; \quad \underline{Q}_{1,M+3} = \underline{\Psi}^\top; \\
\quad \underline{Q}_{2,2} = \underline{\Gamma}^\top A_0 + A_0^\top \underline{\Gamma}  +  \underline{\Xi}_0; \quad \underline{Q}_{2,j+2}= \underline{\Gamma}^\top A_j +H_j^\top \underline{\Upsilon}_{0,j} -A_0^\top \underline{\Omega}_j; \\
\quad \underline{Q}_{2,M+3}= \underline{\Gamma}^\top; 
\quad \underline{Q}_{j+2,j+2}= -\underline{\Omega}_j^\top A_j -A_j^\top\underline{\Omega}_j +\underline{\Xi}^j, \\
\underline{Q}_{s+2,z'+2}=-\underline{\Omega}_s^\top A_{z'} -A_s^\top\underline{\Omega}_{z'} + H_s \underline{\Upsilon}_{s,z'} H_{z'}^\top, \\
\underline{Q}_{j+2,M+3} = -\underline{\Omega}_j^\top; \quad \underline{Q}_{M+3,M+3} = \underline{\gamma} I_n. 
\end{gather*}
\begin{gather*}
j,j'\in\overline{1,M};\; z\in\overline{j+1,M};  \; s\in\overline{1,M-1};\;z'\in\overline{s+1,M}, 
\end{gather*}  
then  system~\eqref{eq:main_system_Per} is oAS with respect to $(T,x_0,\gamma_0,\delta_1,\delta_2)$.  
\end{thm}

\begin{proof}
This proof follows the arguments of Theorem~\ref{theorem:ASTS_main_thm} and Corollary~\ref{cor_AS}, except that the Lyapunov functions 
\begin{gather*}
\overline{V}(x)=x^{\top} \left( \overline{P}_1 + C^\top \overline{P}_2 C \right) x + 2\sum_{j=1}^{M}\sum_{i=1}^{k_j} \overline{\Lambda}_{j}^{i}\int_{0}^{H_j^{(i)} x}f_{j}^{i}(\tau)d\tau, \\ 
\underline{V}(x)=x^{\top}  \left( \underline{P}_1 + C^\top \underline{P}_2 C \right) x +2\sum_{j=1}^{M}\sum_{i=1}^{k_j} \underline{\Lambda}_{j}^{i}\int_{0}^{H_j^{(i)} x}f_{j}^{i}(\tau)d\tau
\end{gather*}
and the conditions 
\begin{gather*}
\overline{\alpha}_{1}(\rVert Cx \rVert) \leq \overline{V}(x)\leq \overline{\alpha}_{2}(\rVert x \rVert),  \\
\underline{\alpha}_{1}(\rVert {C} x \rVert) \leq \underline{V}(x)\leq \underline{\alpha}_{2}(\rVert Cx \rVert),\; \forall x \in \R^n
\end{gather*}
are considered. The descriptor method is also utilized. 
\end{proof}


Note also that the conditions in Theorem~\ref{theorem:oAS} are much milder than the ones in Theorem~\ref{theorem:ASTS_main_thm}, which allows the corresponding LMIs in Theorem~\ref{theorem:oAS} to be solved more easily. The condition $Cx_0\ne0$ can be avoided by allowing $\underline\gamma$ to change its sign.

\section{Application to Recurrent Neural Networks} \label{sec:example}
In this section, we  illustrate the usefulness of the proposed conditions for continuous-time recurrent neural networks (CTRNNs). 

The considered CTRNN is \cite{Benoit-Marand2006}:
\begin{gather}
\begin{aligned}
\dot{\chi}(t) & =  A \chi(t) + W_0 g(W_1 \chi(t)) + u(t), \quad  t\in \R_{+}, \\
y(t) & = \tilde{C} \chi(t),  \label{main_continuous_time_RNN_dynamics}
\end{aligned}
\end{gather} 
where $\chi(t) \in \mathbb{R}^n$  is the state vector, $\chi(0)=\chi_0$; $A\in \R^{n\times n}, W_0 \in \R^{n \times N}, W_1 \in \R^{N \times n}$ are the weight matrices; $g: \R^N \to \R^N$ is the activation function; $u\in\mathscr{L}_{\infty}^{n}$ is the input; $y(t)  \in \mathbb{R}^p$ is the output and $\tilde{C} \in \R^{p \times n}$, $N\geq n$ is the number of neurons in the hidden layer. It is clear that the CTRNN~\eqref{main_continuous_time_RNN_dynamics} is in the form of~\eqref{eq:main_system_Per}, and the Hopfield neural network \cite{Hopfield1984} is a special example of \eqref{main_continuous_time_RNN_dynamics}.

Here we consider typical cases of non-polynomial activation functions
\cite{Apicella2019}, \emph{e.g.}, rectified linear unit,
sigmoid, and hyperbolic tangent functions, satisfying the sector
condition of Assumption~\ref{assu:main}.


During the training process, it is of primary importance to analyze AS and oAS of a trained CTRNN~\eqref{main_continuous_time_RNN_dynamics} at each step. The following corollary is for studying AS of~\eqref{main_continuous_time_RNN_dynamics}.  

\begin{cor} 
If the conditions of Corollary~\ref{cor_AS} are satisfied under the substitutions $A_0 \to A$, $M \to 1$, $A_1 \to W_0$, $F_1 \to g$, $H_1 \to W_1$, then the CTRNN~\eqref{main_continuous_time_RNN_dynamics} is AS with respect to $(T,\epsilon_1,\epsilon_2,\gamma_0,\delta_1,\delta_2)$.  
\end{cor}

\begin{proof}
This result is a direct consequence, considering the expressions of~\eqref{eq:main_system_Per}, \eqref{main_continuous_time_RNN_dynamics} and  Corollary~\ref{cor_AS}.
\end{proof}

We then propose the following main theorem of oAS for the CTRNN~\eqref{main_continuous_time_RNN_dynamics}. 

\begin{thm} \label{thm_CTRNN_oAS}
If the conditions of Theorem~\ref{theorem:oAS}  are satisfied under the substitutions $A_0 \to A$, $M \to 1$, $A_1 \to W_0$, $F_1 \to g$, $H_1 \to W_1$, $C \to \tilde{C}$, then the CTRNN~\eqref{main_continuous_time_RNN_dynamics} is oAS with respect to $(T,x_0,\gamma_0,\delta_1,\delta_2)$.  
\end{thm}

\begin{proof}
This proof is again straightforward under the consideration of the forms of~\eqref{eq:main_system_Per}, \eqref{main_continuous_time_RNN_dynamics} and Theorem~\ref{theorem:oAS}.
\end{proof}

\begin{rem}
Note that if there exist non-zero biases in the activation function, then it is possible to rewrite the system~\eqref{main_continuous_time_RNN_dynamics} as follows: 
\begin{gather*}
\begin{aligned}
\dot{\xi}(t) &=  \tilde{A} \xi(t) + \tilde{W}_0 \tilde{g} \left( \begin{bmatrix} W_1 & b_0 \end{bmatrix} \xi(t) \right) + \tilde{u}(t), \\
y(t) &=  \hat{C} \xi(t)
\end{aligned}
\end{gather*}
which is again in the form of \eqref{eq:main_system_Per}, where 
\begin{gather*}
\xi(t) = \begin{bmatrix} \chi(t) \\ \eta \end{bmatrix} \in \mathbb{R}^{n+N}, \; \tilde{A} = \diag(A,O_{N \times N}),   \\
\tilde{g} \left( \tau \right) = \begin{bmatrix} g(\tau) \\g(\tau) \end{bmatrix},\; \hat{C} = \begin{bmatrix} \tilde{C} & O_{ p \times N} \end{bmatrix}, \\
\tilde{W}_0 = \diag(W_0,O_{N \times N}), \;\tilde{u}(t) = \begin{bmatrix} u(t) \\ O_{N \times 1} \end{bmatrix}. 
\end{gather*}
Then a similar analysis can be performed under the guarantee that $\chi(T)$ or $y(T)$ will be in $\cl(B_{n+N}(\delta_2) \setminus B_{n+N}(\delta_1))$ or $\cl(B_{p}(\delta_2) \setminus B_{p}(\delta_1))$, respectively.
\end{rem}

\begin{rem}
The proposed conditions can also be applied to numerous variations of the CTRNNs, \emph{e.g.}, a variable activation function (VAF) sub-network scheme was introduced in \cite{Apicella2019}, which can be embedded into a full connected CTRNN, resulting in a CTRNN taking the form of~\eqref{eq:main_system_Per} with a sufficiently large $M$.
\end{rem}


\begin{figure}[htb]
\centering \includegraphics[width=0.7\textwidth,height=7cm]{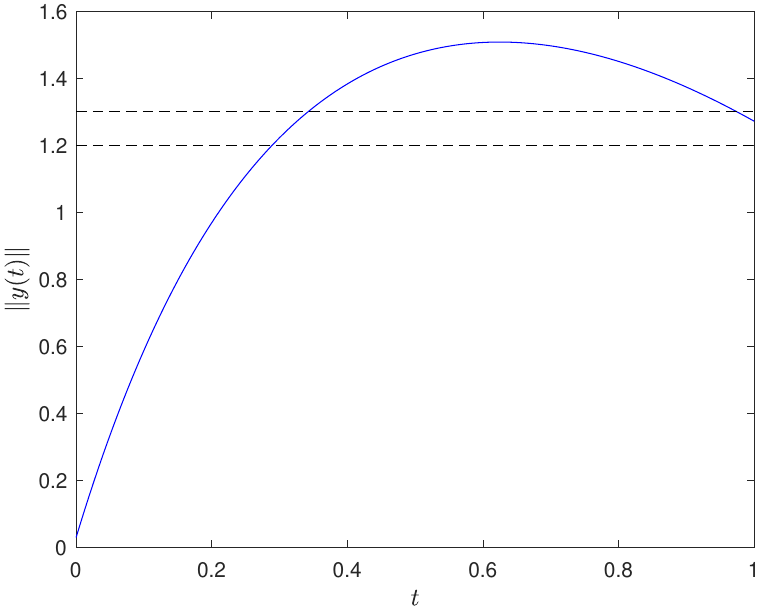}\caption{The trajectory of $\rVert y(t) \rVert$ \label{fig:y}}
\end{figure}

\subsection{Numerical example}

In this subsection, we consider Hopfield neural networks
\cite{Hopfield1984} with one hidden layer, which has $n$ recurrent
nodes of bipolar sigmoids and $N$ neurons in the hidden layer. 

\begin{exmp}
Let $n=2$, $N=50$, and the activation function $g^{\ell} =
\tanh$ for all $\ell \in \overline{1,N}$. The weight matrix $A$ and
$\tilde{C}$ are
\begin{gather*}
A = \begin{bmatrix} -2 & 0  \\0 & -5  \end{bmatrix}, \;
\tilde{C} = \begin{bmatrix} -2.2 & 1.3  \end{bmatrix}.
\end{gather*}
We also obtained the values of $W_0, W_1$ in the Hopfield neural network randomly and analyzed the oAS property. 
Assumption~\ref{assu:main} is fulfilled due to the form of $g$. Also, Assumption~\ref{assu:integral_bound} is satisfied for 
\begin{gather*}
\kappa_{0,1}^{1}=...=\kappa_{0,1}^{50}=0, \;  \kappa_{1,11}^{1}=...=\kappa_{1,11}^{50}=0,\\
\kappa_{2,11}^{1}=...=\kappa_{2,11}^{50}=0; \\ 
\eta_{0,1} = \Lambda_1, \;  \eta_{1,11}^{1}=...=\eta_{1,11}^{50}=0,\\
\eta_{2,11}^{1}=...=\eta_{2,11}^{50}=0 
\end{gather*}
since 
\begin{gather*}
2 {\bf{1}}_{N}^\top \Lambda_1 \begin{bmatrix}  \int_{0}^{W_1^{(1)} \chi}g^{1}(\tau)d\tau \\ \vdots \\ \int_{0}^{W_1^{(N)} \chi }g^{N}(\tau)d\tau \end{bmatrix} = 2 {\bf{1}}_{N}^\top \Lambda_1 \begin{bmatrix}  \ln(\cosh(W_1^{(1)} \chi)) \\ \vdots \\ \ln(\cosh(W_1^{(N)} \chi))  \end{bmatrix}  \\
\leq  {\bf{1}}_{N}^\top \Lambda_1\begin{bmatrix}  (W_1^{(1)} \chi)^2 \\ \vdots \\ (W_1^{(N)} \chi)^2   \end{bmatrix} \leq  \chi^\top W_1^\top \eta_{0,1} W_1 \chi.
\end{gather*}
The oAS conditions in Theorem~\ref{thm_CTRNN_oAS} with $\; \delta_1 =1.2, \; \delta_2=1.3, \; u(t) = \begin{bmatrix}  3\cos(t) \\ \sin(t) \end{bmatrix}, \; T=1$, and $\gamma_0 = 16$ 
are verified. The trajectory of the norm of the output with a nonzero initial state on $t \in [0,1]$  is shown in Fig.~\ref{fig:y}.    From this plot, we see that  the  given  estimates  are  rather
tight, considering the trajectory closely approaches the border at $t=1$, which means oAS of the considered system.  

Moreover, we illustrated the AS property of the considered CTRNN. We omitted the output $y$ and let the other system parameters  remain unchanged for brevity. The AS conditions are verified by 
\begin{gather*}
 \epsilon_1 = 0.09, \; \epsilon_2 =0.1, \; \delta_1 =0.9, \; \delta_2=1, \; u(t) = \begin{bmatrix}  3\cos(t) \\ \sin(t) \end{bmatrix}, \; T=1, \gamma_0 = 16, \\
\beta_1 =100, \; \Lambda = \underline{\Lambda}_1= 0, \; \underline{P} >0,\; -\underline{\Xi}^0 \geq \beta_1 \underline{P}, \; \underline{\Xi}^1 = O_{50 \times 50}, \; -\underline{\Upsilon}_{0,1}  \geq 0, \\
 \underline{\gamma} \leq \frac{ e^{100}\underline{\alpha}_1(0.09) - \underline{\alpha}_2(0.9)}{2.56 \times \left[ e^{100} -1 \right]}, \\
 {\small \underline{Q} = \underline{Q}^\top = \begin{bmatrix}
     -\underline{\Psi}^\top - \underline{\Psi} & &  \underline{\Psi}^\top A + \underline{P}  - \underline{\Gamma}  & &   \underline{\Omega}_1  + \underline{\Psi}^\top W_0  & &   \underline{\Psi}^\top \\
    * & & \underline{\Gamma}^\top A + A^\top \underline{\Gamma}  +  \underline{\Xi}_0 & & \underline{\Gamma}^\top W_0 + W_1^\top \underline{\Upsilon}_{0,1} -A^\top \underline{\Omega}_1 & & \underline{\Gamma}^\top \\ 
    * & & * & & - \underline{\Omega}_1^\top W_0 - W_0^\top\underline{\Omega}_1 & &  -\underline{\Omega}_1^\top \\
    * & & * & & * & & \underline{\gamma} I_n
 \end{bmatrix} \geq 0,   } \\
 \beta_2 = 0.01, \; \Lambda = \overline{\Lambda}_1= 0, \;  \overline{P} >0, \;\overline{\Xi}^0 = O_{2 \times 2}, \; \overline{\Xi}^1 = O_{50\times 50}, \; \overline{\Upsilon}_{0,1}  \geq 0, \\    
  \overline{\gamma} \leq \frac{\overline{\alpha}_1(1) -  e^{0.01}\overline{\alpha}_2(0.1)}{25600  \times \left[ e^{0.01} -1 \right]}, \\ 
  {\small \overline{Q} = \overline{Q}^\top = \begin{bmatrix}
     -\overline{\Psi}^\top - \overline{\Psi} & &  \overline{\Psi}^\top A + \overline{P}  - \overline{\Gamma}  & &   \overline{\Omega}_1  + \overline{\Psi}^\top W_0  & &   \overline{\Psi}^\top \\
    * & & \overline{\Gamma}^\top A + A^\top \overline{\Gamma}  & & \overline{\Gamma}^\top W_0 + W_1^\top \overline{\Upsilon}_{0,1} -A^\top \overline{\Omega}_1 & & \overline{\Gamma}^\top \\ 
    * & & * & & - \overline{\Omega}_1^\top W_0 - W_0^\top\overline{\Omega}_1 & &  -\overline{\Omega}_1^\top \\
    * & & * & & * & & -\overline{\gamma} I_n
 \end{bmatrix} \leq 0.   } 
\end{gather*}
Also, the trajectory of $\rVert x(t) \rVert$ with a nonzero initial state over the time interval $[0,1]$ is presented in Fig. \ref{fig:x}, from which one can easily see that the ASTS is not satisfied since $0.9 \leq \rVert x(t) \rVert \leq 1$ does not hold for all $t \in [0,1]$, while the AS property was verified (only the time instants $t=0$ and $t=1$ are taken into account).  This demonstrates the main difference between ASTS and AS, and the usefulness of the latter notion (in the presence of output, oAS can be applied, for instance, in this example).

\begin{figure}[htb]
\centering \includegraphics[width=0.7\textwidth,height=7cm]{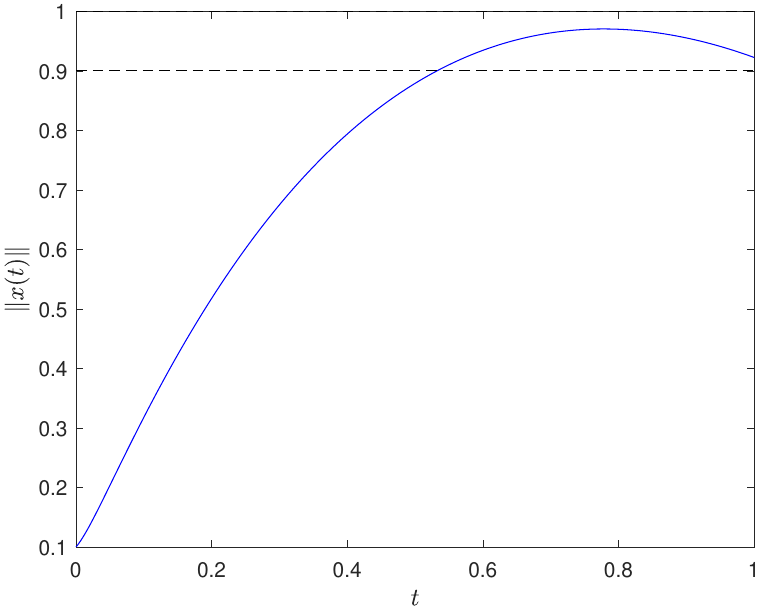}\caption{The trajectory of $\rVert x(t) \rVert$ \label{fig:x}}
\end{figure}

\end{exmp}

\begin{exmp}
In order to further illustrate the validity of the proposed conditions for different neural networks in practice, we then deal with a larger Hopfield neural network with $n=10$ recurrent nodes. Under the settings of $N=50$ and $g^{\ell} =
\tanh$ for $\ell \in \overline{1,N}$, we chose the tuple of matrices $(A,W_0,W_1,\tilde{C})$ in~\eqref{main_continuous_time_RNN_dynamics}
arbitrarily, then the LMIs of
Theorem~\ref{thm_CTRNN_oAS}  are again verified for $\; \epsilon_1 = 0.05, \; \epsilon_2 = 0.1, \; \delta_1 = 2, \; \delta_2 =2.3,  \; T=1$, and $\gamma_0 = 16$. It is worth mentioning that the verification of the corresponding LMIs is feasible due to the high flexibility of the tuning parameters.
\end{exmp}

\section{Conclusion}
In this paper, notions of annular settling (AS) and output annular
settling (oAS) were proposed, and sufficient conditions for annular
short-time stability (ASTS), short-time boundedness with nonzero
initial state (STBNZ), AS, and oAS in generalized Persidskii systems
were given. The formulated conditions were obtained in the form of
linear algebraic and matrix inequalities. Therefore, they can be
constructively verified. An application to a continuous-time recurrent
neural network was presented to validate the proposed results. 


\section*{Appendix. Proof of Theorem~1}
\begin{proof}
Consider a candidate Lyapunov function
\begin{gather*}
\overline{V}(x)=x^{\top} \overline{P} x+2\sum_{j=1}^{M}\sum_{i=1}^{k_j} \overline{\Lambda}_{j}^{i}\int_{0}^{H_j^{(i)} x}f_{j}^{i}(\tau)d\tau.  
\end{gather*}

Under condition~\eqref{overline_V_estimate}, there exist some
functions $\overline{\alpha}_{1}$, $\overline{\alpha}_{2}$ from class
$\K_{\infty}$ such that
\begin{equation*}
\overline{\alpha}_{1}(\rVert x \rVert) \leq \overline{V}(x)\leq \overline{\alpha}_{2}(\rVert x \rVert),\; \forall x \in \R^n
\end{equation*}
due to Finsler's Lemma and Assumption~\ref{assu:main}. Then consider the time derivative of $\overline{V}$ along trajectories of~\eqref{eq:main_system_Per}, for which the following upper estimate holds true under the restriction~\eqref{overline_Q_restriction}:
\begin{eqnarray*}
\dot{\overline{V}} & =& \nabla \overline{V}(x) \dot{x} \\
& = & \dot{x}^{\top} \overline{P} x+x^{\top} \overline{P} \dot{x}+2\sum_{j=1}^{M} \dot{x}^\top  H_j^\top \overline{\Lambda}_{j}F_{j}(H_jx) \\ 
& &+2 \left( \sum_{j=1}^{M} F_j(H_jx)^\top \overline{\Omega}_j^\top -\dot{x}^\top \overline{\Psi}^\top-x^\top \overline{\Gamma}^\top \right)\Bigg(\dot{x} - A_{0}x  \\
& &-\sum_{j=1}^{M}A_{j}F_{j}(H_j x)- u \Bigg) \\
& = & \left[\begin{array}{c}
\dot{x} \\
x\\
F_{1}(H_1 x)\\
\vdots\\
F_{M}(H_M x)\\
u
\end{array}\right]^{\top}\overline{Q}\left[\begin{array}{c}
\dot{x} \\
x\\
F_{1}(H_1 x)\\
\vdots\\
F_{M}(H_M x)\\
u
\end{array}\right]-x^{\top}\overline{\Xi}^{0}x\\
 &  & -\sum_{j=1}^{M}F_{j}(H_jx)^{\top} \overline{\Xi}^{j}F_{j}(H_j x)-2\sum_{j=1}^{M}x^{\top} H_j^\top\overline{\Upsilon}_{0,j}F_{j}(H_j x)\\
 &  & -2\sum_{s=1}^{M-1}\sum_{z=s+1}^{M}F_{s}(H_s x)^{\top} H_s\overline{\Upsilon}_{s,z} H_z^\top F_{z}(H_z x)\\ 
 & &+ \overline{\gamma} u^{\top} u \\
 & \leq & -x^{\top}\overline{\Xi}^{0}x -2\sum_{s=1}^{M-1}\sum_{z=s+1}^{M}F_{s}(H_s x)^{\top}H_s \overline{\Upsilon}_{s,z} H_z^\top F_{z}(H_z x) \\
 &  & -\sum_{j=1}^{M}F_{j}(H_jx)^{\top} \overline{\Xi}^{j}F_{j}(H_j x) -2\sum_{j=1}^{M}x^{\top} H_j^\top \overline{\Upsilon}_{0,j}F_{j}(H_j x) \\ 
 & & +\overline{\gamma} u^{\top} u.
\end{eqnarray*} 
Here the descriptor method \cite{Fridman2014book} was applied on the second step. 

Now we have to show that there exists $\beta_2 \in \mathbb{R}$ such that  
\begin{gather*}
\beta_2 \overline{V}(x) \geq - x^{\top}\overline{\Xi}^{0}x-\sum_{j=1}^{M}F_{j}(H_jx)^{\top} \overline{\Xi}^{j}F_{j}(H_j x)\nonumber \\
-2\sum_{j=1}^{M}x^{\top} H_j^\top\overline{\Upsilon}_{0,j}F_{j}(H_j x)-2\sum_{s=1}^{M-1}\sum_{z=s+1}^{M}F_{s}(H_s x)^{\top}H_s \overline{\Upsilon}_{s,z} H_z^\top F_{z}(H_z x),
\end{gather*}
which is true under Assumption~\ref{assu:integral_bound} and the conditions~\eqref{overline_V_estimate}, \eqref{eq:beta_two_condi}. Therefore, we have
\begin{gather*}
\dot{\overline{V}} \leq \beta_{2} \overline{V}+ \overline{\gamma} u^\top u,
\end{gather*}
so that
\begin{eqnarray*}
\overline{V}(x(t)) 
& \leq &
e^{\beta_{2} t} \overline{V}(x_0) + \int_{0}^t e^{\beta_{2} (t-s)} \overline{\gamma} u(s)^\top u(s) ds\\
& \leq &
e^{\beta_{2} t} \overline{\alpha}_2(\rVert x_0 \rVert) + \frac{\overline{\gamma} \gamma_0^2}{\beta_2} \left( e^{\beta_2 t} -1 \right)\\ 
& \leq &
e^{\beta_{2} t} \overline{\alpha}_2(\epsilon_2) + \frac{\overline{\gamma} \gamma_0^2}{\beta_2} \left( e^{\beta_2 t} -1 \right).
\end{eqnarray*}
It further holds that

\begin{eqnarray*} 
\sup_{t\in[0,T]}\overline{V}(x(t)) &\leq & \sup_{t\in [0,T]} \Big\{ a
\eta(t) + b[\left( \eta(t)-1 \right) \Big\} \\ & = &\begin{cases} a \eta(T) + b
  \left( \eta(T)-1 \right), & \text{if}\; \beta_2 > 0,\\ a +
  \overline{\gamma} \gamma_0^2 T, & \text{if}\; \beta_2 = 0,\\ a
  \eta(0) + b \left( \eta(0)-1 \right), & \text{if}\; \beta_2 <0, a+b \geq 0,
\\ a \eta(T) + b \left( \eta(T)-1\right), & \text{if}\; \beta_2 <0, a+b < 0
\end{cases} 
\end{eqnarray*}

with the settings of 
\begin{gather*}
\eta(t)= e^{\beta_2 t},\; b = \frac{\overline{\gamma} \gamma_0^2}{\beta_2}, \; a = \overline{\alpha}_2(\epsilon_2).
\end{gather*}
Therefore, under conditions~\eqref{gamma_overline_condi}, we obtain
\begin{eqnarray*}
\overline{V}(x(t))\leq  e^{\beta_{2} t} \overline{\alpha}_2(\epsilon_2) + \frac{\overline{\gamma} \gamma_0^2}{\beta_2} \left[ e^{\beta_2 t} -1 \right] 
\leq  \overline{\alpha}_{1}(\delta_2), \; \forall t \in [0,T],
\end{eqnarray*}
by which we see that 
\begin{gather*}
\overline{\alpha}_1(\rVert x(t) \rVert) \leq \overline{V}(x(t)) \leq \overline{\alpha}_1(\delta_2)\quad \Rightarrow \quad  \rVert x(t) \rVert \leq \delta_2, \; \forall t \in [0,T]. 
\end{gather*}

Now let us evaluate the infimum of $\rVert x(t) \rVert$ on  $t\in [0,T]$, for which another candidate Lyapunov function
\begin{gather*}
\underline{V}(x)=x^{\top} \underline{P} x+2\sum_{j=1}^{M}\sum_{i=1}^{k_j} \underline{\Lambda}_{j}^{i}\int_{0}^{H_j^{(i)} x}f_{j}^{i}(\tau)d\tau   
\end{gather*}
is considered. In a similar way as in the previous development, it
follows that there exist some functions $\underline{\alpha}_{1}$, $
\underline{\alpha}_{2} \in \K_{\infty}$ such that
\begin{eqnarray*}
 \underline{\alpha}_{1}(\rVert x \rVert) \leq \underline{V}(x)\leq \underline{\alpha}_{2}(\rVert x \rVert),\; \forall x \in \R^n 
\end{eqnarray*}
due to the condition~\eqref{underline_V_estimate}, Finsler's Lemma and Assumption~\ref{assu:main}. Calculating the time derivative of $\underline{V}$ on trajectories of~\eqref{eq:main_system_Per}, due to~\eqref{underline_Q_restriction} and using the descriptor method, we obtain:

\begin{eqnarray*}
\dot{\underline{V}} & \geq & -x^{\top}\underline{\Xi}^{0}x
-2\sum_{s=1}^{M-1}\sum_{z=s+1}^{M}F_{s}(H_s x)^{\top}H_s
\underline{\Upsilon}_{s,z} H_z^\top F_{z}(H_z x) \\ & &
-\sum_{j=1}^{M}F_{j}(H_jx)^{\top} \underline{\Xi}^{j}F_{j}(H_j x)
-2\sum_{j=1}^{M}x^{\top} H_j^\top \underline{\Upsilon}_{0,j}F_{j}(H_j
x) \\ & & -\underline{\gamma} u^{\top} u.
\end{eqnarray*}
Applying similar arguments, there exists $\beta_1 \in \mathbb{R}$ such that
\begin{gather*}
\beta_1 \underline{V}(x) \leq - x^{\top}\underline{\Xi}^{0}x-\sum_{j=1}^{M}F_{j}(H_jx)^{\top} \underline{\Xi}^{j}F_{j}(H_j x) \nonumber \\
-2\sum_{j=1}^{M}x^{\top} H_j^\top\underline{\Upsilon}_{0,j}F_{j}(H_j x)-2\sum_{s=1}^{M-1}\sum_{z=s+1}^{M}F_{s}(H_s x)^{\top}H_s \underline{\Upsilon}_{s,z} H_z^\top F_{z}(H_z x)
\end{gather*}
due to the condition \eqref{eq:beta_one_condi} and Assumption~\ref{assu:integral_bound}. Then, it can be shown that under the conditions \eqref{gamma_underline_condi}, the property
\begin{eqnarray*}
 e^{\beta_{1} t} \underline{\alpha}_1 (\epsilon_1)
-  \frac{\underline{\gamma} \gamma_0^2}{\beta_1} \left[ e^{\beta_1 t} -1 \right] \geq \underline{\alpha}_2(\delta_1), \quad \forall t \in [0,T]
\end{eqnarray*}
holds, which deduces 
\begin{gather*}
\underline{\alpha}_2(\delta_1) \leq \underline{V}(x(t)) \leq \underline{\alpha}_2(\rVert x(t) \rVert) \quad \Rightarrow \quad  \rVert x(t) \rVert \geq \delta_1, \; \forall t \in [0,T]. 
\end{gather*}
This completes the proof.
\end{proof}


\bibliographystyle{ieeetr}
\bibliography{ieeeconf}

\begin{thebibliography}{10}

\bibitem{Vidyasagar:81:Springer}
M.~Vidyasagar, {\em Input-output analysis of large-scale interconnected
  systems}, vol.~29 of {\em Lecture Notes in Control and Information Sciences}.
\newblock Berlin: Springer-Verlag, 1981.
\newblock Decomposition, well-posedness and stability.

\bibitem{Schaft:96:Springer}
A.~van~der Schaft, {\em {$L\sb 2$}-gain and passivity techniques in nonlinear
  control}, vol.~218 of {\em Lecture Notes in Control and Information
  Sciences}.
\newblock London: Springer-Verlag London Ltd., 1996.

\bibitem{Khalil2002nonlinear}
H.~Khalil, {\em Nonlinear systems}.
\newblock Upper Saddle River, NJ: Prentice-Hall, 2002.

\bibitem{Vidyasagar2002}
M.~Vidyasagar, {\em {Nonlinear Systems Analysis}}.
\newblock Society for Industrial and Applied Mathematics, 2002.

\bibitem{Dorato}
P.~Dorato, ``{An Overview of Finite-Time Stability},'' in {\em {Current Trends
  in Nonlinear Systems and Control}}, pp.~185--194, Birkh{\"{a}}user Boston,
  2006.

\bibitem{dorato1961short}
P.~Dorato, {\em Short-time stability in linear time-varying systems}.
\newblock Polytechnic Institute of Brooklyn, 1961.

\bibitem{Amato2010}
F.~Amato, M.~Ariola, and C.~Cosentino, ``Finite-time stability of linear
  time-varying systems: Analysis and controller design,'' {\em IEEE
  Transactions on Automatic Control}, vol.~55, no.~4, pp.~1003--1008, 2010.

\bibitem{Garcia2009b}
G.~Garcia, S.~Tarbouriech, and J.~Bernussou, ``Finite-time stabilization of
  linear time-varying continuous systems,'' {\em IEEE Transactions on Automatic
  Control}, vol.~54, no.~2, pp.~364--369, 2009.

\bibitem{Amato2009a}
F.~Amato, R.~Ambrosino, M.~Ariola, and C.~Cosentino, ``Finite-time stability of
  linear time-varying systems with jumps,'' {\em Automatica}, vol.~45, no.~5,
  pp.~1354--1358, 2009.

\bibitem{Amato2001a}
F.~Amato, M.~Ariola, and P.~Dorato, ``Finite-time control of linear systems
  subject to parametric uncertainties and disturbances,'' {\em Automatica},
  vol.~37, no.~9, pp.~1459--1463, 2001.

\bibitem{Yang2009}
Y.~Yang, J.~Li, and G.~Chen, ``{Finite-time stability and stabilization of
  nonlinear stochastic hybrid systems},'' {\em Journal of Mathematical Analysis
  and Applications}, vol.~356, no.~1, pp.~338--345, 2009.

\bibitem{Amato2016a}
F.~Amato, G.~{De Tommasi}, A.~Mele, and A.~Pironti, ``New conditions for
  annular finite-time stability of linear systems,'' in {\em {2016 IEEE 55th
  Conference on Decision and Control (CDC)}}, pp.~4925--4930, IEEE, 2016.

\bibitem{Amatoa}
F.~Amato, M.~Ariola, C.~Cosentino, C.~Abdallah, and P.~Dorato, ``Necessary and
  sufficient conditions for finite-time stability of linear systems,'' in {\em
  {Proceedings of the 2003 American Control Conference}}, vol.~5,
  pp.~4452--4456, IEEE, 2003.

\bibitem{Barbashin1961}
E.~Barbashin, ``On construction of {Lyapunov} functions for nonlinear
  systems,'' in {\em {Proc. 1st IFAC World Congress}}, pp.~742--751, 1961.

\bibitem{Persidskii1969}
S.~Persidskii, ``Concerning problem of absolute stability,'' {\em Automation
  and Remote Control}, pp.~5--11, 1969.

\bibitem{Ferreira2005}
L.~Ferreira, E.~Kaszkurewicz, and A.~Bhaya, ``{Solving systems of linear
  equations via gradient systems with discontinuous righthand sides:
  Application to LS-SVM},'' {\em IEEE Transactions on Neural Networks},
  vol.~16, pp.~501--505, 2005.

\bibitem{Mei2021b}
W.~Mei, D.~Efimov, and R.~Ushirobira, ``{On input-to-output stability and
  robust synchronization of generalized Persidskii systems},'' {\em IEEE
  Transactions on Automatic Control}, 2021.

\bibitem{Hsu1987}
H.~Liu and L.~Colvara, ``The inclusion of automatic voltage regulators in power
  system transient stability analysis using {L}yapunov functions,'' {\em IFAC
  Proceedings Volumes}, vol.~20, pp.~1--6, 1987.

\bibitem{Efimov2019a}
D.~Efimov and A.~Y. Aleksandrov, ``{Robust stability analysis and
  implementation of Persidskii systems},'' in {\em {2019 IEEE 58th Conference
  on Decision and Control (CDC)}}, pp.~6164--6168, IEEE, 2019.

\bibitem{Mei2020a}
W.~Mei, D.~Efimov, and R.~Ushirobira, ``Feedback synchronization in
  {Persidskii} systems,'' {\em IFAC-PapersOnLine}, vol.~53, no.~2,
  pp.~2880--2884, 2020.

\bibitem{Mei2020b}
W.~Mei, D.~Efimov, and R.~Ushirobira, ``Towards state estimation of
  {Persidskii} systems,'' in {\em {2020 59th IEEE Conference on Decision and
  Control (CDC)}}, pp.~2881--2886, IEEE, 2020.

\bibitem{Mei2021}
W.~Mei, D.~Efimov, R.~Ushirobira, and A.~Aleksandrov, ``Convergence conditions
  for {Persidskii} systems,'' in {\em {19th European Control Conference}},
  2021.

\bibitem{Guiver2020a}
C.~Guiver and H.~Logemann, ``A circle criterion for strong integral
  input-to-state stability,'' {\em Automatica}, vol.~111, p.~108641, 2020.

\bibitem{Moreno2012}
J.~A. Moreno and M.~Osorio, ``Strict {L}yapunov functions for the
  super-twisting algorithm,'' {\em IEEE Transactions on Automatic Control},
  vol.~57, no.~4, pp.~1035--1040, 2012.

\bibitem{Efimov2021}
D.~Efimov and A.~Aleksandrov, ``On analysis of {P}ersidskii systems and their
  implementations using {LMI}s,'' {\em Automatica}, 2021.

\bibitem{mei2022nonlinear}
W.~Mei, R.~Ushirobira, and D.~Efimov, ``On nonlinear robust state estimation
  for generalized persidskii systems,'' {\em Automatica}, vol.~142, p.~110411,
  2022.

\bibitem{Benoit-Marand2006}
F.~Beno{\^{i}}t-Marand, L.~Signac, T.~Poinot, and J.-C. Trigeassou,
  ``{Identification of non linear fractional systems using continuous time
  neural networks},'' {\em IFAC Proceedings Volumes}, vol.~39, no.~11,
  pp.~402--407, 2006.

\bibitem{Hopfield1984}
J.~J. Hopfield, ``{Neurons with graded response have collective computational
  properties like those of two-state neurons.},'' {\em Proceedings of the
  National Academy of Sciences}, vol.~81, no.~10, pp.~3088--3092, 1984.

\bibitem{Apicella2019}
A.~Apicella, F.~Isgr{\`{o}}, and R.~Prevete, ``{A simple and efficient
  architecture for trainable activation functions},'' {\em Neurocomputing},
  vol.~370, pp.~1--15, 2019.

\bibitem{Fridman2014book}
E.~Fridman, {\em {Introduction to Time-Delay Systems}}.
\newblock Systems {\&} Control: Foundations {\&} Applications, Cham: Springer
  International Publishing, 2014.

\end{thebibliography}

\end{document}